\documentclass[twocolumn]{article}

\usepackage[utf8]{inputenc}
\usepackage{amsmath}
\usepackage{amssymb}
\usepackage{algorithm2e}
\usepackage{chemformula}
\usepackage{todonotes}
\usepackage{longtable}
\usepackage{lscape} 
\usepackage{soul}
\usepackage{amsthm}
\usepackage{multicol}
\usepackage{geometry}

\newtheorem{theorem}{Theorem}
\theoremstyle{definition}
\newtheorem{definition}[theorem]{Definition}
\newtheorem{example}{Example}

\newcommand{\QQ}{\mathbb{Q}}
\newcommand{\RR}{\mathbb{R}}

\begin{document}

\title{A Graph Theoretical Approach for Testing Binomiality of Reversible Chemical Reaction Networks}
\author{
\begin{tabular}[t]{c@{\extracolsep{3em}}c} 
 Hamid Rahkooy  & Cristian Vargas Montero \\
  \small{CNRS, Inria, and the University of Lorraine} &
  \small{CNRS, Inria, and the University of Lorraine}\\ 
 \small{Nancy, France} & \small{Nancy, France} \\
 % \small{hamid.rahkooy@inria.fr}
  \small{MPI Informatics, Saarbrücken, Germany}                 &
                                  \small{cristian.vargas-montero6@etu.univ-lorraine.fr}
  \\
 \small{hrahkooy@mpi-inf.mpg.de} & 
 \end{tabular}
}
\date{}

\maketitle
\begin{abstract}
  We study binomiality of the steady state ideals of chemical reaction
  networks. Considering rate constants as indeterminates, the concept
  of unconditional binomiality has been introduced and an algorithm
  based on linear algebra has been proposed in a recent work for
  reversible chemical reaction networks, which has a polynomial time
  complexity upper bound on the number of species and reactions. In
  this article, using a modified version of species--reaction graphs,
  we present an algorithm based on graph theory which performs by
  adding and deleting edges and changing the labels of the edges in
  order to test unconditional binomiality. We have implemented our
  graph theoretical algorithm as well as the linear algebra one in
  Maple and made experiments on biochemical models. Our experiments
  show that the performance of the graph theoretical approach is
  similar to or better than the linear algebra approach, while it is
  drastically faster than Gr\"obner basis and quantifier elimination
  methods.
\end{abstract}

%=============================================
\section{Introduction}

Chemical reactions are transformations between chemical species where
a creation or elimination of species may happen with respect to
changes in time, pressure, temperature, etc. A set of chemical
reactions is called a chemical reaction network (CRN) and if all the
reactions in a CRN are reversible, it is called a reversible chemical
reaction network (RCRN). We assume \textit{mass--action kinetics} in
this article. The following is an example of an RCRN.
\begin{center}\label{eq:ex1rn}
    \ch{ A + B <=>[ k12 ][ k21 ] C <=>[ k23 ][ k32 ] A + 2 D},
  \end{center}
  with species \ch{A}, \ch{B}, \ch{C} and \ch{D} and \textit{rate
  constants} $k_{12},k_{21},k_{23}$ and $k_{32}$.

  Ordinary differential equations (ODE) can be used to study the
  changes in the concentration of species of a CRN. The ODEs
  associated to the above CRN are

\begin{align}
  \dot{x}_1  &=p_1, & p_1& = -k_{12}x_1x_2+k_{21}x_3+k_{23}x_3& \nonumber\\&&& -k_{32}x_1x_4^{2}&   \nonumber\\
  \dot{x}_2  &=p_2, & p_2& = -k_{12}x_1x_2+k_{21}x_3&  \nonumber\\
  \dot{x}_3  &=p_3, & p_3&= k_{12}x_1x_2-k_{21}x_3-k_{23}x_3\nonumber\\ &&&+k_{32}x_1x_4^{2}&  \nonumber\\
  \dot{x}_4  &=p_4, & p_4& = 2k_{23}x_3-2k_{32}x_1x_4^{2}&\label{eq:n13B4}
\end{align}

The ideal generated by $p_1,p_2.p_3,p_4$ is called the {\em steady
  state ideal} of the CRN. The real positive zeros of the above ideal
are called the \textit{steady states}. Finding steady states is a
fundamental problem in CRN theory. For a thorough introduction to CRN
theory, we refer to Feinberg's Book \cite{Feinberg-Book}.

A CRN is called binomial if its steady state ideal is
binomial. Following the work in \cite{rahkooy2020linear}, in this
article we investigate binomiality over the ring
$\QQ[k_{ij},x_1,\dots,x_n]$, which we call \textit{unconditional
  binomiality}. Some authors have considered binomiality of CRNs over
$\QQ(k_{ij})[x_1,\dots,x_n]$ when $k_{ij}$ are specialised in
$\RR$, e.g., in \cite{perez_millan_chemical_2012}, which we call
\textit{conditional binomiality}.

Binomial ideals and toric varieties are historic topics in
thermodynamic and go back to Boltzmann and Einstein. 
Binomiality and toricity have been widely studied in
mathematics \cite{fulton_introduction_2016,sturmfels_grobner_1996,Eisenbud_1996}. Binomiality
is a hard problem and the typical approach by computing a Gr\"obner
basis is EXPSPACE-complete \cite{mayr_complexity_1982}.

In CRN theory, various articles have been written on binomiality and
toricity, e.g., by Gorban et
al.~\cite{gorban_generalized_2015}, Grigoriev and Weber \cite{GrigorievWeber2012a}
and others \cite{dickenstein2019,sadeghimanesh2019}.
Feinberg \cite{Feinberg1972}
and Craciun, et al. \cite{craciun_toric_2009} have studied toric
dynamical systems.

Dickenstein et al.~have presented sufficient linear algebra conditions
with inequalities for conditional binomiality in 
\cite{perez_millan_chemical_2012}, which lead to MESSI
CRNs \cite{millan_structure_2018}. For homogeneous ideals, it has been
shown that Dickenstein et al.'s condition is necessary as
well \cite{conradi2015detecting}. The latter has been implemented in
Maple and Macaulay II in \cite{MapleCK,alex2019analysis}. A geometric
view towards toricity has been considered by Grigoriev et al.~in 
\cite{grigoriev2019efficiently}, introducing shifted toricity and
presenting algorithms using quantifier
elimination \cite{Davenport:1988:RQE:53372.53374,Grigoriev:88a,Weispfenning:88a}
and Gr\"obner bases \cite{Buchberger:65a,bb-system}
A first order logic test for toricity has been given in 
\cite{sturm2020firstorder}. In \cite{rahkooy2020linear} unconditional
binomiality has been introduced and for reversible reactions a
polynomial time algorithm, based on a linear algebra approach, has
been presented.

Graph theory has been used in the study of CRNs, e.g., for detecting
\textit{concordance} and \textit{weak
  monotonicity} \cite{Feinberg-Book,feinberg2012concordant}.

The main contribution of this article is a graph theoretical approach
for testing unconditional binomiality of an RCRN. We use a modification
of \textit{species--reaction} graphs and present an algorithm
equivalent to the linear algebra approach presented in 
\cite{rahkooy2020linear}. We have implemented the graph theoretical
algorithm as well the linear algebra algorithm in
Maple \cite{maplesoft} and compared them with each other and with the algorithms
based on Gr\"obner basis and quantifier elimination 
presented in \cite{grigoriev2019efficiently} via experiments on
biological models from the BioModels
repository \cite{BioModels2015a}. Our experiments showed that the graph
theoretical and linear algebra approaches are not only drastically
faster, but they can also handle many cases that were not possible to
compute using Gr\"obner basis and quantifier elimination.

The plan of this article is as follows. In Section 
\ref{sec:binomlinal} we briefly review the linear algebra approach in 
\cite{rahkooy2020linear}. Section \ref{sec:binomgraph} presents the
graph theoretical algorithm and the proof of its correctness, which
are the main contributions of this article.  Implementations of and
experiments using both
the graph theoretical and linear algebra algorithms are presented in
Section \ref{sec:implbench}.
    
%=============================================
\section{Testing Unconditional Binomiality via Linear
  Algebra}\label{sec:binomlinal}

In this section we briefly review the linear algebra method in 
\cite{rahkooy2020linear} for testing unconditional binomiality in
RCRNs.

\begin{definition}[Definition 1, \cite{rahkooy2020linear}]\label{def:binom}
  Let \ch{C} be a reversible chemical reaction network with the
  complexes \ch{C_{1}}, \dots ,\ch{C_{s}}, let
  $k_{ij} , 1 \leq i \neq j \leq s,$ be the constant rate of the
  reaction from \ch{C_{i}} to \ch{C_{j}}, and let
  $x_{1}, \cdots , x_{n}$ be the concentrations of the species.
  $b_{ij} :=-k_{ij}m_{i}+k_{ji}m_{j}$ is called the binomial
  associated to the reaction from \ch{C_{i}} to \ch{C_{j}}. If there
  is no reaction between \ch{C_{i}} and \ch{C_{j}}, set $b_{ij} := 0$.
\end{definition}

\begin{example}\label{eq:ex3rn}
  The binomials associated to the reversible reactions presented in
  the introduction section are
  $b_{12} = -k_{12}x_1x_2+k_{21}x_3$ and
  $b_{23} = -k_{23}x_3+k_{32}x_1x_4^{2}$.
\end{example}

Following Definition \ref{def:binom}, one can write the corresponding
ODEs for an RCRN as
\begin{equation}\label{eq:sumb}
  \dot{x}_{k} = p_{k} = \sum_{j=1}^{s} c_{ij}^{(k)}b_{ij}, \quad k=1, \dots, n,
\end{equation}
where $c_{ij}^{(k)}$ is the difference between the stoichiometric
coefficients of the $k-$th species in the reaction \ch{C_i <=> C_j}.
It has been shown that if rate constants are indeterminates then the
above sum of binomials representation is
unique \cite{rahkooy2020linear} .

\begin{definition} [Definition 6, \cite{rahkooy2020linear}] Let \ch{C}
  be an RCRN as in Definition \ref{def:binom} and assume that $p_i$,
  the generators of its steady state ideal, are written as sum of
  binomials as in Equation \ref{eq:sumb}. We define the binomial
  coefficient matrix of \ch{C} to be the matrix whose rows are labeled
  by $p_{1}, \dots , p_{n}$ and whose columns are labeled by nonzero
  $b_{ij}$ and the entry in row $p_k$ and column $b_{ij}$ is
  $c_{ij}^{(k)}$.
\end{definition}

The binomial coefficient matrix can be used to test the unconditional
binomiality of an RCRN.
\begin{theorem}[Theorem 8, \cite{rahkooy2020linear}]\label{th:binom}
  A RCRN is unconditionally binomial (i.e., assuming the rate
  constants to be indeterminates) if and only if the reduced row
  echelon form(RREF) of its binomial coefficient matrix has at most one
  nonzero entry at each row.
\end{theorem}

The theorem leads to a linear algebra approach for testing
unconditional binomiality \cite[Algorithm 1]{rahkooy2020linear},
which is polynomial time in terms of the size of matrix, which is the
maximum of the number of species and reactions in the RCRN. In Section
\ref{sec:implbench}, we present an implementation of the algorithm in
Maple and compare it with the other approaches.

% ==================================
\section{A Graph Theoretical Approach}\label{sec:binomgraph}

In this section we present a graph theoretical algorithm equivalent to
the linear algebra approach in \cite{rahkooy2020linear}. A CRN can be
represented as a graph in several ways. A first idea is the
well--known complex--reaction graph of a CRN, in which the complexes
are considered as the vertices and the reactions as the directed
edges. Another idea is \textit{species--reaction graphs}, which is
used to study \textit{concordance} and \textit{weakly monotonicity of
  kinetics} in a CRN \cite[Theorem 10.5.5, Theorem 11.5.1, Theorem
11.6.1]{Feinberg-Book}.
We use a modified version of the species--reaction graph, defined only
for RCRNs, for detecting unconditional binomiality. For definition and
notations on graph theory we rely on \cite{west2017introduction}.

\begin{definition}[Modified Species--Reaction Graph of an RCRN]
  Let $S$ be the set of species and $R$ the set of reactions of a
  given RCRN. The modified species--reaction graph $G$ of the RCRN is
  defined as follows.
\begin{itemize}
\item For each species $s \in S$, consider a vertex of $G$ (species
  vertices)
\item For each reaction $r \in R$, consider a vertex of $G$(reactions
  vertices)
\item For each reaction vertex, there exists an undirected edge to the
  species vertices of the species that appear in the reaction.
\item Each edge of the graph is labeled by an integer number which is
  the difference between the stoichiometric coefficients of the
  species(present at one end of the edge) in the reactant and product
  complexes.
\end{itemize}
\end{definition}

\begin{example}\label{ex:rnbin}
  The species--reaction graph of the following RCRN is showed in
  Figure \ref{fig:graph2}.
\begin{center}
    \ch{2 A + B  <=> C <=> A<=> 2 B}.    
\end{center}
\end{example}

\begin{figure}
\centering
\includegraphics[scale=0.8]{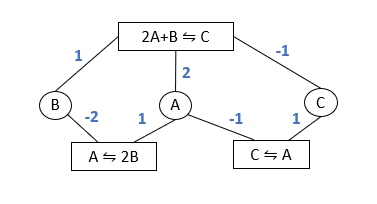}
\caption{Mod. Spec--Reac. Graph (Example \ref{ex:rnbin})}
\label{fig:graph2}
\end{figure}

In order to check unconditional binomiality of an RCRN using the
modified species--reaction graph, we present an algorithm that
modifies the graph by adding and deleting the edges, and updating the
adge labelsso that unconditional binomiality can be easily read off
from the final graph. The algorithm simulates the procedure described
in Theorem \ref{th:binom} from Section \ref{sec:binomlinal} and
Algorithm 1 from \cite{rahkooy2020linear} by reducing a binomial
coefficient matrix to its RREF in order to test unconditional
binomiality of the network.

The idea of the algorithm is as follows. First, we iterate through the
reaction vertices, selecting and marking an (random) unmarked species
vertex connected to the current reaction vertex (initially all species
vertices are unmarked). If there are no unmarked species vertices,
then the current reaction vertex is skipped. Then, we delete all the
edges incident to the current reaction vertex that are different from
the edge going to the marked species vertex. Furthermore,
if an edge exists from the marked species
vertex to a reaction vertex, then we add the edge from the reaction
vertex to the current species vertex. Nevertheless, if the edges
already exists, then we update the label accordingly and if the new
label is zero then we eliminate the edge. The final graph is reached
when all reaction vertices have been visited. At the end of the
algorithm, unconditional binomiality is checked by testing if each
component of the final graph contains either only one species vertex
or one species vertex and one reaction vertex. The algorithm is fully
described below.

\begin{algorithm}
 \footnotesize
 \caption{Testing unconditional binomiality via
   graphs \label{alg:binom-graph}}
  \DontPrintSemicolon
  \SetAlgoVlined
   \LinesNumbered
  
  \SetKwProg{Fn}{Function}{}{end}
  \SetKwFunction{BINOMTEST}{BinomialityTestViaGraph}
  \SetKwFunction{CREATEGRAPH}{CreateGraph}
  \SetKwFunction{MRK}{SetMark}
  \SetKwFunction{GS}{GetConnectedSpecies}
  \SetKwFunction{GR}{GetConnectedReactions}
  \SetKwFunction{ISNMRK}{IsNotMarked}
  \SetKwFunction{GCO}{GetCoeff}
  \SetKwFunction{ICO}{UpdCf}
  \SetKwFunction{ELED}{ElimEdge}
  \SetKwFunction{ADDED}{AddEdge}
  \SetKwFunction{CONN}{AreConnected}
  \SetKwFunction{IB}{IsUnconditionallyBinomial}
  \SetKwFunction{GETSVS}{GetSpeciesVertices}
  \SetKwFunction{GETRVS}{GetReactionVertices}

  \Fn{\BINOMTEST{$\mathcal{S,R,G}$}}{
    \KwIn{\\$\mathcal{S}$: set of species of the RCRN\\$\mathcal{R}$: set of reactions of the RCRN}
    \KwOut{UnconditionallyBinomial or NotUnconditionallyBinomial}
    $\mathcal{G} := \CREATEGRAPH{$\mathcal{S,R}$}$\\
    $SV := \GETSVS{$\mathcal{G}$}$\\
    $RV := \GETRVS{$\mathcal{G}$}$\\
    \ForEach{$s \in SV$}{
        $\MRK{s,false} $
    }
    \ForEach{$r \in RV$}{
        $speciesFound := false$\\
        $cs := null$\\
        $rSpecies := \GS{$\mathcal{G},r$}$\\
        \ForEach{$sr \in rSpecies$}{
            \uIf{\ISNMRK{$sr$}}{
                $speciesFound := true$\\
                $cs := sr$\\
                $break$
            }
        }
        \uIf{$speciesFound$}{
            $\MRK{cs,true} $\\
            \ForEach{$s^{\prime} \in rSpecies$}{
                \uIf{$s^{\prime} \neq cs$}{
                    $multX :=\frac{-\GCO{$\mathcal{G}$,$s^{\prime},r$}}{\GCO{$\mathcal{G}$,cs,r}}$\\
                    $\ELED{$\mathcal{G},s^{\prime},r$}$\\
                    $sReactions := \GR{$\mathcal{G},cs$}$\\
                    \ForEach{$r^{\prime} \in sReactions$}{
                        \uIf{$r^{\prime} \neq r$}{
                            \uIf{$\CONN{$\mathcal{G},s^{\prime},r^{\prime}$}$}{
                                $cf := \GCO{$\mathcal{G},cs,r^{\prime}$}*multiX + \GCO{$\mathcal{G},s^{\prime},r^{\prime}$}$\\
                                \uIf{$cf \neq 0$}{
                                    $\ICO{$\mathcal{G},s^{\prime},r^{\prime},cf$}$\\
                                }
                                \Else{
                                    $\ELED{$\mathcal{G},s^{\prime},r^{\prime}$}$\\
                                }
                            }
                            \Else{
                              $cf := \GCO{$\mathcal{G},cs,r^{\prime}$}*multiX$\\ $\ADDED{$\mathcal{G},s^{\prime},r^{\prime},cf$}$ 
                            }
                        }
                    }
                    
                }
      
      }
    }

    }
    \uIf{\IB{$\mathcal{G}$}}{
      $R:=UnconditionallyBinomial$\;
    }
    \Else{$R:=NotUnconditionallyBinomial$\;
    }
    \Return{$R$}\;
  }
  \end{algorithm}

  The functions in Algorithm \ref{alg:binom-graph} are
  self-explanatory. We just mention that \texttt{ElimEdge} eliminates
  the edge connecting a species vertex and reaction vertex, and
  \texttt{UpdCf} updates the coefficient of the edge that goes from a
  species vertex to a reaction vertex (Detailed description of the
  functions can be found in the Git repository \cite{MapleModules}).

\begin{theorem}\label{th:graph}
  Algorithm \ref{alg:binom-graph} is correct, terminates and its
  asymptotic worst case time complexity can be bounded by
  $\mathcal{O}(\max(r,n)^\omega)$, where $\omega$ is the constant
  appearing in the complexity of matrix multiplication, $r$ is the
  number of reactions and $n$ is the number of species of an RCRN.
\end{theorem}

\begin{proof}  
  \textbf{(Proof of the correctness)} Assume that the algorithm output
  is unconditionally binomial, then we must prove that the steady
  state ideal of the RCRN is binomial. To do so, we will show that the
  steps of the graph theoretical Algorithm \ref{alg:binom-graph} are
  equivalent to the steps of the linear algebra approach in 
  \cite[Algorithm 1]{rahkooy2020linear} using Reduced Row Echelon Form
  (RREF) in the binomial coefficient matrix. Based on Theorem
  \ref{th:binom}, Algorithm 1 from \cite{rahkooy2020linear} initially
  selects a pivot in the matrix which is equivalent to marking an
  unmarked species vertex that is connected to a reaction vertex in
  steps 9 to 14. Reducing the nonzero entries to zero in the column of
  the selected pivot in the matrix is equivalent to eliminating the
  edges from the reaction vertex to all other species vertices that
  are not the one selected in step 20. While performing the reduction,
  some entries of other rows may be affected. The equivalent to this
  in Algorithm \ref{alg:binom-graph} is the update of the coefficients
  in some edges in step 27 or the elimination of edges in step 29 or
  the addition of edges in step 32. Then, obtaining the RREF of the
  matrix is equivalent to a combination of the following items:
\begin{itemize}
\item species vertices without edges, which are equivalent to the zero
  rows and/or
\item reaction(species) vertices connected to at least one
  species(reaction) vertex, which are equivalent to the final
  columns(rows) of the matrix.
\end{itemize}
Finally, testing unconditional binomiality in the matrix (via checking
that it has at most one nonzero entry at each row) is equivalent to
checking that the components of the final graph contain either:
\begin{itemize}
\item only one species vertex or,
\item one species vertex and one reaction vertex connected to one
  another.
\end{itemize}

\textbf{(Proof of termination)} For the proof of termination, note
that the graph has finitely many edges and vertices. Hence, each of
the loops terminates at some point. The loops at lines 6 and 22
eventually terminate as they iterate through reaction vertices which
are finite and there is no creation of such vertices
anywhere. Likewise, the loops at lines 4,10 and 17 terminate as they
iterate through species vertices which are also finite, and no new
vertices are created.

\textbf{(Proof of the complexity bound)} This comes from the fact that
each operation in the graph theoretical algorithm corresponds to an
operation in the linear algebra approach.
\end{proof}

As a sidenote, for any transformed graph one can generate a binomial
coefficient matrix by taking the species vertices as rows, reactions
vertices as columns and the labels of edges as entries and vice versa,
from any binomial coefficient matrix one can generate a graph by
applying the reverse procedure.

\begin{example}\label{ex:rnnotbin}
  Graphs generated at each step of Algorithm \ref{alg:binom-graph} for
  the following RCRN has been shown in Figure \ref{fig:graph3}.
\begin{center}
    \ch{3 B <=> 2 C + A <=> 2 D + 2 B<=> 3 B}.
\end{center}
The final graph has a component that contains two species vertices and
three reaction vertices and therefore the RCRN is not unconditionally
binomial. On the other hand, Algorithm 1 in \cite{rahkooy2020linear}
gives the following matrix
which shows that the RCRN is not unconditionally binomial since the
first and second rows contain more than one nonzero entries.
\end{example}
\begin{figure}
\centering
\includegraphics[scale=0.6]{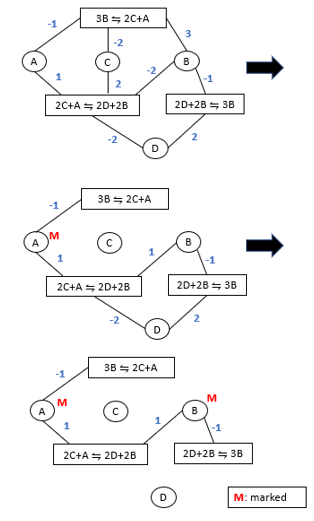}
\caption{Algorithm \ref{alg:binom-graph} on Example
  \ref{ex:rnnotbin}.}
\label{fig:graph3}
\end{figure}

%==========================================
\section{Experiments}\label{sec:implbench}

We have implemented our graph theoretical approach in Algorithm
\ref{alg:binom-graph} in Section \ref{sec:binomgraph} as well as the
linear algebra approach in \cite[Algorithm 1]{rahkooy2020linear} in
Maple. Both algorithms are available in the Git repository
\cite{MapleModules}.
The performance of the implementations is tested on the biomodels from
the BioModels repository \cite{BioModels2015a} and the results are
compared to the experiments done on the same biomodels using Gr\"obner
basis and quantifier elimination in 
\cite{grigoriev2019efficiently}. Note that in order to be able to test
unconditional binomiality of biomodels, we have assumed reversibility
with free reverse rate constants, while in 
\cite{grigoriev2019efficiently} binomiality with preassigned values of
rate constants are tested. The result of the experiments can be found
in the Git repository \cite{MapleModules}.

One can notice that overall, Algorithm \ref{alg:binom-graph} performed
equal or better than \cite[Algorithm 1]{rahkooy2020linear}. For a vast
amount of models the difference of the performance of those two
algorithms is almost zero, which was predictable as the proof of
Theorem \ref{th:graph} suggests that the steps in the graph
theoretical algorithm and the linear algebra approach are
equivalent. However, for some models (e.g., 205, 293 and 574) the
graph theoretical approach is faster. It is worth noting that these
models have large binomial coefficient matrices, which may suggest
that for RCRN with a large number of species and/or reactions the
graph theoretical approach can be faster. This will be investigated
further in future as the number of models with large matrices is not
enough for a through comparison at this stage.

Comparing Algorithm \ref{alg:binom-graph} with the Gr\"obner basis and
quantifier elimination computations in \cite{grigoriev2019efficiently}
shows that for the great majority of the cases, the graph theoretical
approach is much faster.  More importantly, there are twenty biomodels
that Gr\"obner basis and/or quantifier elimination methods in 
\cite{grigoriev2019efficiently} run out of time (for a six hour limit
of computations), while those cases were handled in less than three
seconds via both graph theory and linear algebra approaches. For six
biomodels Gr\"obner basis computations terminate in less than six
hours, however, our graph theory algorithm as well as the linear
algebra approach are at least 1000 times faster than those. For ten
biomodels the graph theory (and the linear algebra) approach is at
least 500 times faster than the computations via Gr\"obner basis and
quantifier elimination.

An interesting observation is the quite high number (sixty nine) of
the biomodels that are not considered in 
\cite{grigoriev2019efficiently} for the unclear numeric values of
their rate constants, whereas our graph theoretical (and linear
algebra) approaches on almost all of those cases terminated in less
than a second.

A comparison of the performance of our graph theory algorithm with the
algorithms proposed in 
\cite{perez_millan_chemical_2012,conradi2015detecting} are similar to
the performance of the linear algebra algorithm in 
\cite{rahkooy2020linear} vs algorithms in 
\cite{perez_millan_chemical_2012,conradi2015detecting}. This is
because, as it is mentioned earlier in this section, the graph
theoretical algorithm has a similar (or better) performance to the
linear algebra algorithm in \cite{rahkooy2020linear}. In particular,
for two reversible biomodels in the database (Biomodels 491 and 492),
the graph theoretical method is more than twice faster than the linear
algebra in \cite{rahkooy2020linear}, which means that it is much
faster than the algorithms in 
\cite{perez_millan_chemical_2012,conradi2015detecting}.  More details
of some comparisons between the linear algebra methods in the
literature for testing unconditional binomiality and conditional
binomiality can be found in \cite[Section 3]{rahkooy2020linear}.

% ==================================
\section{Conclusion}
The present work introduces an efficient graph theoretical algorithm
for testing unconditional binomiality in an RCRN, which is different
from the conditional binomiality considered by several other
authors. The algorithm is essentially equivalent to the linear algebra
approach presented earlier for testing unconditional
binomiality. Implementations of the graph theoretical algorithm as
well the linear algebra approach are done in Maple and experiments are
carried on over several biomodels.

While the graph theoretical algorithm seem to have a slight advantage
over the linear algebra approach, the experiments reveal drastic
advantage of both of those methods over the existing algorithms based
on Gr\"obner basis and quantifier elimination. Additionally, many
cases that could not be handle by the Gr\"obner basis and quantifier
elimination methods in a reasonable time, were tested in less than few
seconds via the graph theoretical approach.
  
%==========================================
\subsection*{Acknowledgments}
This work has been supported by the bilateral project
ANR-17-CE40-0036/DFG-391322026 SYMBIONT 
\cite{BoulierFages:18a,BoulierFages:18b}.

\bibliography{arxiv2.bib}
\bibliographystyle{plain}

\end{document}